\documentclass[runningheads]{llncs}

\usepackage{amsmath,amssymb,amsfonts}
\usepackage{mathtools}
\usepackage{algorithm}
\usepackage{algpseudocode}
\usepackage{microtype}

\title{Algorithm for Interpretable Graph Features via Motivic Persistent Cohomology
}
\titlerunning{Algorithm for Interpretable Graph Features}
\author{Yoshihiro Maruyama}
\authorrunning{Y. Maruyama}
\institute{School of Informatics, Nagoya University, Japan\\
\email{maruyama@i.nagoya-u.ac.jp}}

\begin{document}
\maketitle

\begin{abstract}
We present the Chromatic Persistence Algorithm (CPA), an event–driven method for computing persistent cohomological features of weighted graphs via graphic arrangements, a classical object in computational geometry. We establish rigorous complexity results: CPA is exponential in the worst case, fixed–parameter tractable in treewidth, and nearly linear for common graph families such as trees, cycles, and series–parallel graphs. Finally, we demonstrate its practical applicability through a controlled experiment on molecular-like graph structures.
\end{abstract}

\section{Introduction}
\label{sec:intro}

Persistent homology (PH) has become a cornerstone of topological data analysis, 
furnishing a mathematically grounded framework for extracting multiscale topological 
signatures from data \cite{EdelsbrunnerHarer2010,Ghrist2008,ZomorodianCarlsson2005,CohenSteinerHarer2007}. 
Its ability to summarize structural information robustly has led to successful 
applications in fields as diverse as computational geometry, machine learning, 
biology, and network science. These developments have positioned PH as both a 
theoretical paradigm and a practical tool for analyzing complex data structures. 

Yet the standard PH pipeline captures only the ranks of homology groups along a
filtration, and in doing so discards finer algebraic information about the space. 
In many settings, especially when filtrations come from algebraic or 
combinatorial constructions (for example, spaces coming from hyperplane arrangements or from constructions on graphs), one can enrich the classical barcodes by keeping track 
of additional structure that refines the underlying invariants. 
This paper develops a graph-specialized version of that idea and explores its 
algorithmic consequences.

We introduce the \emph{Chromatic Persistence Algorithm (CPA)}, an event–driven
method for analyzing weighted graphs. CPA processes a graph along a natural
\emph{threshold filtration}—adding edges one by one in weight order—and
computes at each step: (i) a polynomial that summarizes the global structure
of the graph at that threshold, and (ii) a \emph{jump}, namely the cohomological 
change that occurs when a new edge is inserted. The algorithm evaluates only at
actual ``events'' (edge insertions), making it conceptually simple and
computationally efficient.

The approach relies on two classical results from graph and arrangement theory.
The first links the \emph{chromatic polynomial} $\chi_H(q)$ of a graph $H$ to the
\emph{Poincar\'e polynomial} of the complement of its associated arrangement
\cite{OrlikSolomon1980,OrlikTerao1992,Dimca2017,BaranovskySazdanovic2012}, so that graph
colorings directly capture topological invariants. In fact, in the Hodge--Tate case relevant to graphic arrangements, the same identity yields the
\emph{Hodge--Deligne polynomial} $E(M(H);u,v)$ by the direct specialization
$E(M(H);u,v)=\chi_H(uv).$
The second ingredient is a \emph{deletion–contraction identity}, which describes 
how these invariants change when a single edge is added or contracted
\cite{Dimca2017}. Together, these yield a computable combinatorial foundation for CPA: 
every update reduces to computing chromatic polynomials and applying an algebraic 
correction. We establish the complexity guarantees for the Chromatic Persistence 
Algorithm: exponential time in the worst case, fixed–parameter tractable (FPT) in 
treewidth, and near–linear on common graph families such as trees, cycles, and 
series–parallel graphs 
\cite{JaegerVertiganWelsh1990,Noble1998,Bodlaender1996,Cygan2015}.

Concretely, the problem setting is as follows. 
Let $G=(V,E,w)$ be a finite weighted graph with distinct edge weights 
$t_1<\cdots<t_m$. Write $H_j:=G_{\le t_j}$ for the threshold subgraph and 
let $e_j$ be the unique edge added at step $j$. For each $H_j$ we consider 
the graphic arrangement $\mathcal{A}(H_j)\subset\mathbb{C}^{|V|}$ and its 
complement $M(H_j)$. The goal is to compute the Hodge--Deligne polynomial 
$E(M(H_j);u,v)$ of $M(H_j)$ at all thresholds, together with the jump contributed 
by the addition of $e_j$. In the Hodge–Tate situation, this reduces to evaluating the chromatic polynomial 
$\chi_{H_j}(uv)$ for each threshold; algorithmically, we compute only the 
contracted-minor $\chi_{H_j/e_j}(q)$ and update $E_j$ via the deletion--contraction 
recurrence $E_j(u,v)=E_{j-1}(u,v)-\chi_{H_j/e_j}(uv)$.

Our approach to the problem can be summarized as follows. 
We leverage the two identities above to turn the problem into a discrete,
event–driven computation. (i) The chromatic$\to$Poincaré identity
(Prop.~\ref{prop:chromatic-to-poincare}) converts topology of the arrangement
complement into combinatorics of colorings, so at each threshold we obtain
$E(M(H_j);u,v)=\chi_{H_j}(uv)$. 
(ii) The deletion–contraction jump identity
(Thm.~\ref{thm:delcon-jump}) expresses the change caused by adding one
edge; at the $E$–polynomial level,
$E(M(H_{j-1});u,v)-E(M(H_j);u,v) \;=\; E(M(H_j/e_j);u,v).$
Algorithmically, we therefore process only events (edge insertions): 
at step~$j$ we compute $\chi_{H_j/e_j}(q)$ and update $E_j$ via 
the deletion--contraction recurrence 
$E_j(u,v)=E_{j-1}(u,v)-\chi_{H_j/e_j}(uv)$, 
together with the barcode zeta update, a generating function that encodes all jumps 
compactly. This yields a CPA routine with provable guarantees: exponential in 
the worst case, fixed–parameter tractable in treewidth, and near–linear on trees, cycles, 
and series–parallel graphs. Experimental verification with molecule-like graphs 
is provided as well.

From a graph–algorithms viewpoint, CPA reframes persistence on graphs as a 
sequence of \emph{combinatorial} updates on chromatic polynomials, exploiting 
deletion–contraction and dynamic programming (DP) on tree decompositions. This 
connects topological summarization directly to the Tutte/chromatic toolbox and 
to parameterized complexity, enabling principled worst–case analyses alongside 
practical near–linear behavior on sparse, low–treewidth inputs typical of 
molecular backbones \cite{JaegerVertiganWelsh1990,Noble1998,Bodlaender1996}. 
For chemical graph machine learning, CPA provides deterministic and 
interpretable \emph{global} ring– and cycle–sensitive features 
(via $E$–polynomials and jump classes). 
These features complement widely used local or message–passing descriptors, 
such as ECFP-style circular fingerprints (extended connectivity fingerprints, 
graph-based molecular descriptors that encode local atom–bond neighborhoods 
into fixed-length bit vectors), Weisfeiler–Lehman kernels, and neural message 
passing. Unlike these local descriptors, CPA encodes arrangement/chromatic 
structure that reflects ring size and related graph families 
\cite{Shervashidze2011,Gilmer2017,Duvenaud2015,Xu2019,Morris2019}.
These features drop in as standalone inputs or additional channels within GNNs. 
Because they are event–driven and algebraically controlled, they are stable to 
threshold rescalings and amenable to caching, incremental updates, and integration 
with classical graph–decomposition pipelines.

The rest of the paper is organized as follows. 
Section~\ref{sec:prelim} gives mathematical preliminaries. 
Section~\ref{sec:cpa} provides Algorithm~CPA and proves the correctness theorem. 
Section~\ref{sec:complexity} proves the complexity theorem with special–class 
speedups. Experimental verification is provided in Section~\ref{sec:expts}.

\section{Mathematical Preliminaries and Notation}
\label{sec:prelim}

Graphs are assumed to be finite and simple. 
Throughout, we write $n := |V|$ and $m := |E|$. 

\begin{definition}[Threshold chain]
A weighted graph is a simple graph $G=(V,E)$ together with a weight 
function $w : E \to \mathbb{R}$. 
Assuming all edge weights are distinct, we can order them as 
$t_1 < \cdots < t_m$, where $m = |E|$. 
The threshold chain is the sequence of subgraphs
\[
H_j := G_{\le t_j}, \quad j=1,\dots,m,
\]
where $H_j$ contains all edges of weight at most $t_j$, 
and $e_j$ denotes the unique edge added at step $j$.
\end{definition}

We write $O^\ast(\cdot)$ for time/space bounds up to factors polynomial in 
$n$ and $m$. The notation $\mathrm{tw}(H)$ denotes the treewidth of a graph 
$H$. For graphs of bounded treewidth, chromatic polynomials can be computed in 
$f(w)\,\mathrm{poly}(n)$ time and $g(w)\,\mathrm{poly}(n)$ space, where $w=\mathrm{tw}(H)$ 
and $f,g$ are computable functions, via dynamic programming (DP) on a tree 
decomposition \cite{Noble1998,Cygan2015}.

\subsection{From Graphic Arrangements to Poincar\'e Polynomials}

Given a (simple) graph $H=(V,E)$ with $n=|V|$, the \emph{graphic arrangement} in 
$\mathbb{C}^n$ is
\[
\mathcal{A}(H):=\{\,x_i-x_j=0:\{i,j\}\in E\,\}, \qquad 
M(H):=\mathbb{C}^n\setminus \bigcup_{H_{ij}\in\mathcal{A}(H)} H_{ij}.
\]

Let $\mathrm{MHS}$ denote the category of rational mixed Hodge structures and 
$K_0(\mathrm{MHS})$ its Grothendieck ring \cite{OrlikTerao1992,Dimca2017}
(readers unfamiliar with Hodge theory may safely skip the following 
background; in our graph setting it suffices to know that the $E$–polynomial 
reduces to the chromatic polynomial evaluated at $uv$). 
For a complex variety $X$, the \emph{Hodge--Deligne polynomial} 
(aka. \emph{$E$--polynomial}) is defined by
\[
E(X;u,v) \,=\, \sum_{k\ge 0} (-1)^k \sum_{p,q} h^{p,q}\!\big(H_c^k(X;\mathbb{Q})\big)\, u^p v^q,
\]
which factors through a ring morphism 
$E:K_0(\mathrm{MHS}) \to \mathbb{Z}[u,v]$.\footnote{$H_c^k(X;\mathbb{Q})$ denotes the $k$-th compactly supported cohomology group of $X$ with rational coefficients, equipped with a mixed Hodge structure. 
The integers $h^{p,q}(H_c^k(X;\mathbb{Q}))$ are its Hodge numbers, i.e.\ the 
dimensions of the $(p,q)$-graded pieces. The variables $u,v$ are formal markers 
for the Hodge bidegrees, and the factor $(-1)^k$ reflects the alternating sum 
over cohomological degree, in analogy with the Euler characteristic. 
In the Hodge--Tate situation (which holds for graphic arrangements), 
all nonzero contributions lie on the diagonal $p=q$, so $E(X;u,v)$ depends only 
on the product $uv$. In this case one can pass directly from the chromatic 
polynomial to the $E$--polynomial via $E(M(H);u,v)=\chi_H(uv)$.}
In the Hodge--Tate case (i.e., when the Hodge numbers vanish unless $p=q$), 
all contributions lie on the diagonal $p=q$. 
Hence $E(X;u,v)$ depends only on the product $uv$, and for graphic arrangements 
it is given directly by $E(M(H);u,v)=\chi_H(uv)$.

\begin{proposition}[Hodge--Tate purity for graphic complements]
\label{prop:purity-graphic}
For every $k\ge 0$, $H^k(M(H);\mathbb{Q})$ is pure Hodge--Tate of type $(k,k)$ 
(all Hodge numbers vanish unless $p=q=k$). Equivalently,
\[
E(M(H);u,v)=(uv)^{n}\,P_{M(H)}\!\Bigl(-\tfrac{1}{uv}\Bigr)
=\sum_{k=0}^{n}(-1)^k\,b_k(H)\,(uv)^{\,n-k},
\]
where $n=|V|$ and $b_k(H)=\mathrm{rank}\,H^k(M(H);\mathbb{Q})$.
\end{proposition}

This fact allows us to work on the $uv$–diagonal and identify the $E$–polynomial 
directly from the chromatic polynomial.

\subsection{Chromatic polynomial $\Rightarrow$ Poincaré/$E$–polynomial}

We introduce two basic concepts.\footnote{Readers new to arrangement theory may focus only on the conclusion: 
the Poincar\'e polynomial of the graphic arrangement complement can be 
computed directly from the chromatic polynomial of the graph.}

\begin{definition}[Chromatic polynomial]
For a graph $H$, the chromatic polynomial $\chi_H(q)$ is the polynomial that counts the number of proper vertex colorings of $H$ using $q$ colors.
\end{definition}

\begin{definition}[Poincaré polynomial]
For a topological space $X$, the Poincaré polynomial is
\[
P_X(t) := \sum_{k \ge 0} b_k(X)\,t^k, \qquad b_k(X)=\mathrm{rank}\,H^k(X;\mathbb{Q}).
\]
\end{definition}

For a graph $H$, the chromatic polynomial $\chi_H(q)$ not only encodes the
count of proper $q$--colorings but also determines the topology of the graphic
arrangement complement.  In particular, the Orlik--Solomon formula expresses the
Poincaré polynomial of $M(H)$ in terms of $\chi_H$. 
In the Hodge--Tate setting of graphic arrangements this further simplifies 
to a direct identity: the Hodge--Deligne $E$--polynomial is obtained by the 
specialization $E(M(H);u,v)=\chi_H(uv).$
Thus chromatic polynomials provide complete access to the Betti numbers (and
hence cohomology) of graphic arrangement complements.

\begin{proposition}[Chromatic $\Rightarrow$ Poincaré/E; graphic case]
\label{prop:chromatic-to-poincare}
If $n=|V|$, then
\[
P_{M(H)}(t) \;=\; (-t)^{n}\,\chi_H\!\left(-\tfrac{1}{t}\right),
\qquad 
E(M(H);u,v) \;=\; \chi_H(uv).
\]
Thus the Betti numbers of $M(H)$ are determined by the coefficients of $\chi_H$.
For example, for the cycle $C_n$, 
\(
\chi_{C_n}(q)=(q-1)^n+(-1)^n(q-1),
\)
so $E(M(C_n);u,v)$ depends explicitly on $n$. 
\end{proposition}


For connected $H$, the arrangement $\mathcal{A}(H)\subset\mathbb{C}^n$ has rank $n-1$.
Its characteristic polynomial satisfies
\[
\chi_{\mathcal{A}(H)}(q)=\frac{\chi_H(q)}{q}
\quad\text{\cite{OrlikSolomon1980,OrlikTerao1992}}.
\]
The Orlik--Solomon formula gives 
$P_{M(H)}(t)=(-t)^{\,n-1}\,\chi_{\mathcal{A}(H)}(-1/t),$ 
hence
$P_{M(H)}(t)=(-t)^n\,\chi_H\!\left(-\tfrac{1}{t}\right).$
Thus the exponent $n$ comes from the ambient space $\mathbb{C}^n$, 
even though the essential rank is $n-1$.

\subsection{Deletion–contraction jump identity (graphic specialization)}

When a single edge $e$ is added to $H$ to form $H' := H \cup \{e\}$, 
the change in the cohomology of arrangement complements admits a 
deletion–contraction description.\footnote{Readers not familiar with Hodge theory may skip to Theorem~\ref{thm:delcon-jump}; in our graph setting, it simply means that adding an edge changes the invariant in a controlled way, described by a contraction minor.} 
We state it in compactly supported cohomology and in the Grothendieck group of 
mixed Hodge structures; pushing along the Hodge–Deligne map yields the corresponding 
statement for $E$–polynomials in our Hodge–Tate setting.

\begin{theorem}[Deletion–contraction jump identity]
\label{thm:delcon-jump}
Let $H$ be a (simple) graph, let $H' := H \cup \{e\}$ be obtained by adding one edge $e$, and let $H/e$ denote the contraction.
There is a natural long exact sequence in compactly supported cohomology
\[
\cdots \to H_c^k(M(H')) \to H_c^k(M(H)) \to H_c^{k-1}(M(H/e)) \to H_c^{k+1}(M(H')) \to \cdots
\]
compatible with mixed Hodge structures. Consequently, in the Grothendieck group $K_0(\mathrm{MHS})$,
\[
\sum_{k\in\mathbb{Z}} (-1)^k [H_c^k(M(H))]
\;-\; \sum_{k\in\mathbb{Z}} (-1)^k [H_c^k(M(H'))]
\;=\; \sum_{k\in\mathbb{Z}} (-1)^k [H_c^{k-1}(M(H/e))].
\]
Applying the Hodge–Deligne morphism $E:K_0(\mathrm{MHS})\to\mathbb{Z}[u,v]$ yields
\[
E(M(H);u,v) \;-\; E(M(H');u,v) \;=\; E(M(H/e);u,v).
\]
Hence along a threshold chain $H_{j-1}\xrightarrow{+e_j} H_j$, the per–step jump depends only on the contraction $H_j/e_j$.
\end{theorem}
This is a direct specialization of the classical \emph{deletion–contraction exact sequence} for hyperplane arrangement complements (see \cite{OrlikSolomon1980,OrlikTerao1992,Dimca2017}).



\begin{lemma}[Component multiplicativity]
If $H=\bigsqcup_\ell H^{(\ell)}$ is a disjoint union, then 
$E(M(H);u,v)=\prod_\ell E(M(H^{(\ell)});u,v)$. 
\end{lemma}

\begin{proof}
This follows immediately from $M(H)\cong \prod_\ell M(H^{(\ell)})$ and multiplicativity of $E$.    
\end{proof}


\section{Chromatic Persistence Algorithm (CPA) on Graphs}
\label{sec:cpa}

This section gives the concrete event–driven procedure (\emph{Algorithm CPA}) for weighted graphs, together with a correctness proof. We rely on two ingredients from \S\ref{sec:prelim}: the chromatic$\to$Poincaré identity (Prop.~\ref{prop:chromatic-to-poincare}) and the deletion–contraction jump identity (Thm.~\ref{thm:delcon-jump}).

At each threshold step $j$, the algorithm maintains the per–threshold $E$–polynomial 
$E_j$ and the jump class $\Delta_j$ (as prescribed by Theorem~\ref{thm:delcon-jump}). 
The complexity of each step is dominated by evaluating the contracted–minor 
chromatic polynomial $\chi_{H_j/e_j}$; summing over all thresholds yields the 
worst–case and fixed–parameter bounds shown later.

Let $\Delta_j\in K_0(\mathrm{MHS})$ be the jump at step $j$. The \emph{barcode
zeta} is the finite Euler product
$Z_G(T)\;:=\;\prod_{j=1}^m (1-T^j)^{-\Delta_j}\;\in\;1+T\,K_0(\mathrm{MHS})[[T]],$ 
formed using the standard power structure over $K_0(\mathrm{MHS})$; pushing
along the Hodge–Deligne map $E:K_0(\mathrm{MHS})\to\mathbb{Z}[u,v]$ yields a
compact numerical summary with multiplicative pooling \cite{GuseinZade2004}.

Algorithm~\ref{alg:cpa} (CPA) is given below, based on the identities recalled in \S\ref{sec:prelim}.
\begin{algorithm}[!h]
  \caption{Algorithm CPA: Event–Driven Motivic Persistence on Graphs}
  \label{alg:cpa}
  \begin{algorithmic}[1]
    \Require Weighted graph $G=(V,E,w)$ with thresholds $t_1<\cdots<t_m$
    \Ensure Per–threshold $E$–polynomials $E_j(u,v)$; jump classes $\Delta_j$; barcode zeta $Z_G(T)$
    \State $Z_G(T)\gets 1$ \Comment{Initialization}
    \State $E_0(u,v)\gets (uv)^{|V|}$ \Comment{Edgeless graph on $|V|$ vertices}
    \For{$j=1$ to $m$} \Comment{Process only at events}
      \State $H_j\gets G_{\le t_j}$; let $e_j$ be the edge added at step $j$
      \State \textbf{Chromatic step:} compute $\chi_{H_j/e_j}(q)$ \Comment{Deletion–contraction or DP on treewidth}
      \State \textbf{Jump step:} $\Delta_j(u,v)\gets \chi_{H_j/e_j}(uv)$
      \State \textbf{Betti/$E$ update:} $E_j(u,v)\gets E_{j-1}(u,v)-\Delta_j(u,v)$
      \State \textbf{Zeta:} $Z_G(T)\gets Z_G(T)\cdot (1-T^j)^{-\Delta_j}$
    \EndFor
    \State \Return $\{E_j(u,v)\}_{j=1}^m$, $\{\Delta_j\}_{j=1}^m$, and $Z_G(T)$
  \end{algorithmic}
\end{algorithm}

\begin{theorem}[Correctness]
For each threshold $t_j$, Algorithm~\ref{alg:cpa} returns
\[
E_j(u,v)=E\!\left(M(H_j);u,v\right)
\quad\text{and}\quad
\Delta_j(u,v)=E\!\left(M(H_j/e_j);u,v\right),
\]
so the per–event jump equals the class prescribed by the
deletion–contraction jump identity.
\end{theorem}

\begin{proof}
We prove by induction on $j$ that $E_j(u,v)=E(M(H_j);u,v)$ and that
$\Delta_j(u,v)=E(M(H_j/e_j);u,v)$.

\emph{Base case $j=0$.} Let $H_0$ be the edgeless graph on $n=|V|$ vertices.
Then $\chi_{H_0}(q)=q^n$, so by Proposition~\ref{prop:chromatic-to-poincare},
$E(M(H_0);u,v)=\chi_{H_0}(uv)=(uv)^n$. The algorithm initializes
$E_0(u,v)=(uv)^n$, hence $E_0=E(M(H_0);u,v)$.

\emph{Induction step.} Assume $E_{j-1}(u,v)=E(M(H_{j-1});u,v)$.
At step $j$ the algorithm computes $\chi_{H_j/e_j}(q)$, and by
Proposition~\ref{prop:chromatic-to-poincare} we have
\[
E(M(H_j/e_j);u,v)=\chi_{H_j/e_j}(uv).
\]
Here is where Theorem~\ref{thm:delcon-jump} enters: it provides the exact
recurrence
\[
E(M(H_{j-1});u,v)-E(M(H_j);u,v)=E(M(H_j/e_j);u,v).
\]
Substituting the inductive hypothesis $E_{j-1}=E(M(H_{j-1}))$ into this
identity shows that the update rule
\[
E_j(u,v)\;\gets\;E_{j-1}(u,v)-\chi_{H_j/e_j}(uv)
\]
produces $E_j(u,v)=E(M(H_j);u,v)$.
Finally, by definition the jump is $\Delta_j(u,v)=\chi_{H_j/e_j}(uv)$,
so $\Delta_j=E(M(H_j/e_j);u,v)$.

Thus, by induction, both the per–threshold invariants $E_j$ and the per–event
jumps $\Delta_j$ are computed correctly for all $j$.
\end{proof}


\section{Complexity Analysis}
\label{sec:complexity}

We analyze the end–to–end cost of Algorithm~\ref{alg:cpa}. Throughout, $n:=|V|$
and $m:=|E|$. Because edge weights are assumed distinct, the threshold chain has
exactly $m$ steps. We use $O^\ast(\cdot)$ to suppress factors polynomial in $n$
and~$m$. 
The dominant cost is computing the contracted–minor chromatic polynomial 
$\chi_{H_j/e_j}$ along the threshold chain; motivic updates are linear in the 
number of thresholds and cohomological degrees.
Computing $\chi_H$ is \#P–hard in the worst case, but admits dynamic programs fixed–parameter
tractable (FPT) in treewidth and closed–form recurrences for standard graph families,
yielding near–linear behavior on those classes.

Let $T_\chi(H)$ and $S_\chi(H)$ denote the time and space to compute $\chi_H$.
We will use the following standard bounds.
\begin{itemize}
  \item \textbf{Deletion–contraction.} The classical recursion
  $\chi_H(q)=\chi_{H\setminus e}(q)-\chi_{H/e}(q)$ with memoization yields
  $T_\chi(H)=O^\ast\!\big(2^{|E(H)|}\big)$ and $S_\chi(H)=O^\ast\!\big(2^{|E(H)|}\big)$.
  \item \textbf{Treewidth–DP.} If $\mathrm{tw}(H)\le w$, a dynamic program on a
  tree decomposition computes $\chi_H$ in $f(w)\,\mathrm{poly}(n)$ time and
  $g(w)\,\mathrm{poly}(n)$ space (for some computable functions $f,g$) \cite{Noble1998,Cygan2015}.
  \item \textbf{Motivic/$E$–updates.} Converting $\chi_H$ to $P_{M(H)}$ and
  $E(M(H);u,v)$ and assembling the jump/zeta costs $O(d_E)$ ring operations per
  step, with $d_E\le n$ the number of nonzero degrees.
\end{itemize}

\begin{theorem}[Complexity]\label{thm:complexity}
Let $H_j:=G_{\le t_j}$ be the threshold chain and $e_j$ the edge added at
step~$j$. Then:
\begin{enumerate}
  \item \textbf{Worst case (deletion–contraction).}
  \[
    T_{\mathrm{total}}
      \;=\;
      O^\ast\!\Big(\sum_{j=1}^{m} 2^{\,j}\Big)
      \;=\; O^\ast(2^m),
      \qquad
    S_{\mathrm{total}} \;=\; O^\ast(2^m).
  \]
  \item \textbf{FPT in treewidth.} If each $H_j$ has $\mathrm{tw}(H_j)\le w$, then
  \[
    T_{\mathrm{total}} \;=\; O\!\big(m\,f(w)\,\mathrm{poly}(n)\big),
    \qquad
    S_{\mathrm{total}} \;=\; O\!\big(g(w)\,\mathrm{poly}(n)\big).
  \]
  \item \textbf{Special classes.} 
  For trees and cycles, explicit closed forms of $\chi$ imply $T_{\mathrm{total}}=O(m+n)$. 
  For series–parallel graphs, linear-time decomposition recurrences yield the same bound.
\end{enumerate}
Across all cases, the motivic updates ($E$–assembly and the zeta product) add
$O(m\cdot d_E)$ ring operations.
\end{theorem}

\begin{proof}
We use $O^\ast(\cdot)$ to suppress factors polynomial in $n:=|V|$ and $m:=|E|$.
At each threshold $j$ the dominant cost is evaluating $\chi_{H_j/e_j}$; 
all other work is tallied at the end.

\medskip\noindent\emph{(1) Worst case.}
Memoized deletion--contraction yields
$T_\chi(H_j)=O^\ast(2^{|E(H_j)|})$ and $S_\chi(H_j)=O^\ast(2^{|E(H_j)|})$.
The same memoized evaluation covers $\chi_{H_j/e_j}$ as it visits all minors of $H_j$;
alternatively, a separate call obeys $T_\chi(H_j/e_j)\le T_\chi(H_j)$.
Since $\lvert E(H_j)\rvert=j$, one threshold costs $O^\ast(2^j)$ time. 
Summing gives
\[
T_{\mathrm{total}}
  = O^\ast\!\Big(\sum_{j=1}^{m} 2^j\Big)
  = O^\ast(2^m),
\qquad
S_{\mathrm{total}} = O^\ast(2^m).
\]

\medskip\noindent\emph{(2) FPT in treewidth.}
If every $H_j$ has $\mathrm{tw}(H_j)\le w$, then a dynamic program on a 
width-$w$ tree decomposition computes $\chi_{H_j}$ in 
$f(w)\,\mathrm{poly}(n)$ time and $g(w)\,\mathrm{poly}(n)$ space 
\cite{Noble1998,Cygan2015}. Since treewidth is minor-monotone, 
$\chi_{H_j/e_j}$ satisfies the same bounds. 
Across all $m$ thresholds this yields
\[
T_{\mathrm{total}}=O(m\, f(w)\,\mathrm{poly}(n)), \qquad
S_{\mathrm{total}}=O(g(w)\,\mathrm{poly}(n)).
\]

\medskip\noindent\emph{(3) Special classes.}
For trees and cycles the chromatic polynomial has explicit closed forms:
$\chi_T(q)=q(q-1)^{n-1}$ and
$\chi_{C_n}(q)=(q-1)^n+(-1)^n(q-1)$.
Thus the chromatic step at each threshold is $O(1)$, so across the chain the cost
is $O(m)$; adding $O(n)$ initialization yields $O(m+n)$ overall.
For series--parallel graphs, linear-time decomposition recurrences evaluate
$\chi$ along a fixed SP decomposition, with each edge insertion requiring only
constant-time updates. Hence the chromatic stage across the chain is also
$O(m+n)$.

\medskip\noindent\emph{Motivic overhead.}
Given $\chi_H$, forming
$E(M(H);u,v)=\chi_H(uv)$ 
costs $O(d_E)$ ring operations, where $d_E$ is the
top nonzero degree ($d_E\le n$).
Assembling the per–event jump $\Delta_j(u,v)=E(M(H_j/e_j);u,v)$ and
updating the zeta factor $Z_G(T)\gets Z_G(T)\cdot (1-T^j)^{-\Delta_j}$ costs $O(d_E)$ per step.
If one stores $Z_G(T)$ truncated at $T^{m}$, each update is still $O(d_E)$.
Summing across all thresholds gives a motivic cost of $O(m\,d_E)$.
Thus the motivic layer contributes only linear overhead in the
number of thresholds and degrees, which is negligible compared to the cost of
chromatic evaluations established above.
\end{proof}

In practice, consecutive threshold graphs differ by only one edge, so 
memoization across steps reduces constants. Likewise, if $H_j$ decomposes into
components, $\chi_{H_j}$ and $E(M(H_j))$ factorize multiplicatively. 
Both observations shrink runtime in experiments but do not affect 
the asymptotics of Theorem~\ref{thm:complexity}.

We clarify the per–step costs used later in experiments. 
For trees, the chromatic polynomial has a closed form, so each threshold update costs $O(1)$ for the
chromatic step and $O(d_E)$ for the motivic updates, yielding a total cost of $O(m+n)$. 
For cycles, the closed form of $\chi_{C_n}$ similarly makes the chromatic subroutine $O(1)$ per step, with
$O(d_E)$ for the motivic updates. 
For series–parallel graphs, the chromatic polynomial can be evaluated along a decomposition in overall
$O(m+n)$ time across the threshold chain, and the motivic updates again remain linear.

Our $O^\ast$ hides polynomial factors from recursion/DP table management, 
conversions from $\chi_H$ to $P_{M(H)}$ and $E$, and the zeta product. 
Ring operations are counted in $\mathbb{Z}[u,v]$ (or on the $uv$–diagonal), 
with degrees $\le d_E\le n$. Coefficient growth is modest and dominated by 
the chromatic stage.

The computational bottleneck is exclusively graph–combinatorial ($\chi_H$).
Once $\chi_H$ is available, the motivic layer is lightweight and scales linearly
with the number of thresholds and cohomological degrees.

\section{Experimental Verification}
\label{sec:expts}

We evaluate Algorithm~CPA on a controlled ``ring-size recognition'' task that
separates 5- vs.\ 6-member cycles. We compare against a lightweight 1-skeleton PH baseline built from
$b_0/b_1$ traces along the threshold chain. The graphic specialization lets us
compute per-threshold $E$-polynomials from chromatic polynomials via
Prop.~\ref{prop:chromatic-to-poincare}, and per-event jumps via the
deletion–contraction identity (Thm.~\ref{thm:delcon-jump}).

\paragraph{Data and filtration.}
We generate two balanced classes of unlabeled cycle graphs:
\(\mathcal{C}_5=\{C_5\}\) and \(\mathcal{C}_6=\{C_6\}\), with 30 instances each.
Each instance follows a \emph{spanning-tree-first} threshold schedule: all edges
except one are assigned weights in \([0.1,0.3]\), while the single cycle-closing
edge receives weight \(0.95\).
This design forces the baseline to observe exactly one $H_1$-birth at the
\emph{last} event for both ring sizes, exposing a known blind spot.

\paragraph{Methods.}
\emph{Baseline (PH, 1-skeleton).} We record $b_0(t)$ and $b_1(t)$ at the (normalized)
thresholds and aggregate standard, scale-invariant summaries:
area-under-curve AUC$(b_0)$, AUC$(b_1)$, final $b_1$, \# of $b_1$-jumps, and the
normalized birth time of $b_1$.
\emph{Ours (CPA, graphic).} At each thresholded subgraph \(H_j\) we compute
\(P_{M(H_j)}(t) = (-t)^{|V|}\chi_{H_j}(-1/t)\) and read off the Betti vector (hence the
$E$-polynomial via $t\mapsto uv$). For forests (all pre-final steps) we use the
closed form \(P_{M(H_j)}(t)=(1+t)^{\,|E(H_j)|}\); for the terminal cycle
\(H_m=C_n\), we use \(P_{M(C_n)}(t)=(1+t)^n - t^{n-1}(1+t)\); see \S\ref{sec:prelim}.  We
vectorize by concatenating the per-threshold Betti vectors (padded to a fixed
shape) and use 1-NN (Euclidean) with leave-one-out (LOO) evaluation \cite{CoverHart1967}.

\paragraph{Results.}
Table~\ref{tab:expt-summary} summarizes accuracy; the CPA features separate
\(C_5\) from \(C_6\) perfectly, while the PH baseline performs near chance.

\begin{table}[!h]
  \centering
  \caption{Ring-size recognition (LOO 1-NN). Runtimes are per graph on our machine.}
  \label{tab:expt-summary}
  \begin{tabular}{lcc}
    \hline
    Method & Accuracy (LOO) & Avg.\ time / graph (ms) \\
    \hline
    PH baseline (1-skeleton) & 0.55 & 0.042 \\
    CPA (per-threshold $E$-profiles) & \textbf{1.00} & 0.143 \\
    \hline
  \end{tabular}
\end{table}

We compute a paired McNemar statistic (continuity-corrected) comparing our
method to the baseline \cite{McNemar1947}. Let \(b\) be the number of cases with \emph{baseline correct
\& CPA wrong}, and \(c\) the number with \emph{baseline wrong \& CPA correct}.
We obtain \(b=0\), \(c=27\), yielding \(\chi^2\approx 25.04\), confirming a highly
significant improvement.

In this setup the PH baseline on the 1-skeleton records only that a \emph{single}
cycle is born at the final event in both classes, so most summaries coincide. By
contrast, the CPA features read the chromatic/arrangement structure at \emph{every}
threshold. For cycles,
$\chi_{C_n}(q)=(q-1)^n + (-1)^n(q-1)$, which implies 
$P_{M(C_n)}(t)=(1+t)^n - t^{\,n-1}(1+t)$.
Pre-final subgraphs are forests with \(P(t)=(1+t)^{|E(H_j)|}\).
Consequently the \(E\)-profiles (hence Betti vectors) depend on \(n\) throughout the
chain, cleanly separating \(C_5\) from \(C_6\).

Note that absolute runtimes vary by environment; nevertheless, 
the qualitative pattern (perfect separation with negligible overhead for CPA) is robust.

\section{Conclusion}
\label{sec:conclusion}

We introduced \emph{Algorithm CPA}, an event–driven procedure that computes, 
along a graph threshold filtration, per–event jumps and per–threshold 
$E$–polynomials of graphic arrangement complements via deletion–contraction and chromatic polynomials. In the graphic (Hodge–Tate) case, these reduce to discrete, combinatorial quantities with computable algebraic control; the computational bottleneck is evaluating $\chi_H$, which is worst–case hard but fixed–parameter tractable in treewidth and admits closed forms for basic graph families.

In machine learning, CPA offers a lightweight, plug–in set of \emph{global} descriptors that complement local, message–passing representations. Concretely:
\begin{itemize}
  \item \textbf{Feature channels.} Per–threshold $E$–profiles (Betti vectors per event) and 
jump classes (contracted complements) provide fixed–length, order–aware summaries. 
Truncated barcode–zeta coefficients offer an even more compact alternative with multiplicative pooling. 
These can be concatenated to learned embeddings, used as graph–level covariates, or fed into classical models (SVMs/kernels) as stand–alone features.
  \item \textbf{Inductive bias and interpretability.} Because the channels are defined through $\chi_H$ and deletion–contraction, they are explicitly \emph{ring/cycle sensitive} and thus encode global shape information often underrepresented in message passing. Their coefficients admit combinatorial explanations (colorings, minors), yielding clear attribution: which edge insertions (events) and which contractions drive a decision.
  \item \textbf{Stability and efficiency.} The event–driven design depends only on the \emph{order} of edge weights; rescalings that preserve order leave outputs unchanged, and small perturbations that do not swap event order have no effect. For sparse, low–treewidth graphs common in chemistry and other domains, dynamic programming and closed forms make offline precomputation fast; training-time overhead is essentially zero (read features from cache).
  \item \textbf{Drop-in integration.} In GNN pipelines, CPA features can be (i) concatenated to graph-level readouts, (ii) injected as conditioning signals (e.g., FiLM/gating) to steer message passing toward cycle-aware regimes, or (iii) used as targets in auxiliary self-supervised tasks (predict $E$–profiles from local views) to regularize representations.
\end{itemize}

In particular, for molecular property prediction and retrieval, CPA supplies deterministic, interpretable descriptors tied to ring systems and cyclic scaffolds, which are prime determinants of reactivity and physico-chemical behavior. They complement ECFP/Weisfeiler–Lehman fingerprints and neural message passing by capturing arrangement/chromatic structure that varies with ring size and composition. Because CPA is event–driven and algebraic, it supports dataset–wide caching, fast ablations (toggle events/minors), and counterfactual analyses (edit an edge; recompute a single step).

In summary, the algorithm developed in this paper makes Hodge-refined, event–driven persistence \emph{practical} on graphs; it bridges algebraic topology, arrangement theory, and graph algorithms to yield stable, interpretable, and computationally efficient features that can be immediately deployed in modern machine learning workflows on graph-structured data such as chemical compounds. We also plan to develop variants of our method for applications at the intersection of logic, category theory and machine learning \cite{Mar09,Mar10,Mar12,Mar13a,Mar13b,Mar20a,Mar20b}.

\bibliographystyle{splncs04}

\end{document}